\documentclass[letterpaper, 10 pt, conference]{ieeeconf}

\usepackage{booktabs}
\usepackage{fontawesome}
\usepackage{comment}
\usepackage{array}
\usepackage{cite}
\usepackage{amssymb}
\usepackage{graphicx} 
\newcommand{\norm}[1]{\left\lVert#1\right\rVert}
\usepackage{epstopdf}
\usepackage{float}
\usepackage{tikz}
\usepackage{algpseudocode}
\usepackage{algorithm}
\usepackage{balance}
\usepackage[fancythm,fancybb]{jphmacros2e} 
\usepackage{amsfonts}
\usepackage{mathrsfs}
\usepackage{amsmath}
\usepackage{amsthm}
\usepackage{comment}
\usepackage{mathtools}
\usepackage{siunitx}
\usepackage{tikz}
 \usepackage[norefs,nocites]{refcheck}
\usetikzlibrary{shapes,arrows}
\let\labelindent\relax
\usepackage{enumitem}
\usepackage{dsfont}

\usepackage{tabularx}
\usepackage{collcell}

\usetikzlibrary{fit}

\tikzstyle{block} = [draw, fill=blue!20, rectangle, 
	minimum height=3em, minimum width=3em]
	\tikzstyle{joint} = [draw, fill=black, circle, 
	inner sep=0pt, minimum size=2pt]
	\tikzstyle{sum} = [draw, fill=blue!20, circle, node distance=1cm]
	\tikzstyle{input} = [coordinate]
	\tikzstyle{output} = [coordinate]
	\tikzstyle{pinstyle} = [pin edge={to-,thin,black}]

\begin{document}
\IEEEoverridecommandlockouts
	\allowdisplaybreaks

\title{Coordinated UAV Beamforming and Control\\ for Directional Jamming and Nulling}
\author{Filippos Fotiadis$^1$ \qquad\qquad Brian M. Sadler$^1$ \qquad\qquad Ufuk Topcu$^1$
\thanks{$^1$F. Fotiadis, B. M. Sadler and U. Topcu are with the Oden Institute for Computational Engineering and Sciences, The University of Texas at Austin, Austin, TX, 78712, 
USA. Email:
{\tt\small ffotiadis@utexas.edu, brian.sadler@austin.utexas.edu, utopcu@utexas.edu}.}
}
\maketitle

\begin{abstract}
 Efficient mobile jamming against eavesdroppers in wireless networks necessitates accurate coordination between mobility and antenna beamforming.
We study the coordinated beamforming and control problem for a UAV that carries two omnidirectional antennas, and which uses them to jam an eavesdropper while leaving a friendly client unaffected. 
The UAV can shape its jamming beampattern by controlling its position, its antennas' orientation, and the relative phasing for each antenna. We derive a closed-form expression for the antennas' phases that guarantees zero jamming impact on the client. In addition, we determine the antennas’ orientation and the UAV’s position that maximizes jamming impact on the eavesdropper through an optimal control problem, optimizing the orientation pointwise and the position through the UAV’s control input. Simulations show how this coordinated beamforming and control scheme enables directional GPS denial while guaranteeing zero interference towards a friendly direction.
\end{abstract}

\section{Introduction}

Jamming is a widely used technique for both enhancing privacy and performing adversarial attacks in wireless communications and control. It takes place when electromagnetic interference arrives at an antenna, subsequently obstructing or distorting the reception of other communication signals. In some operational settings, agents have employed friendly jamming as a means to defend their assets from hostile drones \cite{ukraine_mirage_ew_2025, dronesex, droneex2}.
Agents also often use jamming maliciously, for example, to interfere with the reception of GNSS signals \cite{insidegnss_mscantonia2024, hump, gpsjamex}. This dual-use nature has motivated the design of control strategies that can operate effectively in the presence of, or in coordination with, jamming.

Two factors largely determine the effectiveness of a jamming signal: the distance to the target and the shape of the antenna’s radiation pattern.
Because signal strength decays with distance, staying close to the target ensures sufficient interference power. At the same time, the antenna’s beampattern, i.e., the way its transmitted energy is distributed over different directions, dictates where that power is delivered. A well-shaped beam can focus interference on an eavesdropper at a known direction, whereas a poorly shaped one may also radiate energy toward friendly receivers. Figure~\ref{fig:beams_example} illustrates this fact, showing how beamshaping can direct interference selectively at an adversary while avoiding spread towards unintended directions.

Beamshaping for directional jamming has been studied extensively in wireless communication and antenna theory; however, its combination with dynamics is relatively limited. For example, \cite{brian2uav} examined beampattern design for static jamming scenarios, while related studies \cite{brian2, wang2016artificial, zhu2014joint, yagiz, kanel} explored different aspects of antenna configuration and beamshaping, all without considering the case where the antenna carrier is a mobile UAV. A related strand incorporates UAV mobility by optimizing trajectories \cite{zhang2017secrecy, li2018uav, zhong2018coop, 9525433}, but typically in slot-based formulations that ignore vehicle dynamics and actuator limits, and can generate beampatterns that radiate interference toward friendly users. These approaches highlight the value of beamshaping and UAV mobility, but leave open the question of how to integrate beamforming with dynamics, control, and nulling constraints in a unified framework.

\begin{figure}[!t]
\centering
\includegraphics[width=1\linewidth]{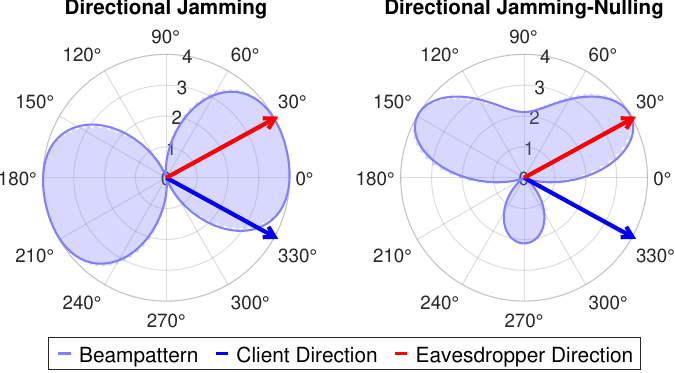}
\caption{\small Examples of antenna beampatterns. Left: the beampattern directs maximum jamming power toward the eavesdropper, but interference also leaks toward the client. Right: the beampattern directs maximum jamming power to the eavesdropper while maintaining a null in the client’s direction.} \label{fig:beams_example}
\end{figure}

Research in control systems, in contrast, has emphasized the role of mobility in performing or defending against jamming attacks; however, the underlying communication model is often simplified and separated from mobility. For example, \cite{bhattacharya2010game, bhattacharya2011spatial, kan2016coverage, valianti2021multi, valianti2024cooperative} presented strategies for planning or agent placement to perform jamming, but assume disk models\footnote{The deterministic disk model assumes that jamming is effective within a fixed radius from the receiver, and is otherwise not effective.}, single omnidirectional antennas, or fixed beampatterns. These provide useful guidelines for the geometry of jammer effectiveness, but restrict capabilities to cases where the target is far from friendly clients or in line of sight; otherwise, jamming the target also jams the friendly client. Other studies have modeled jamming at the control system level \cite{etesami2019dynamic, li2015jamming, lee2015jamming, zhang2016attack, surabhi, hideaki, teixeira, 8693810}, but do not account for the communication model that relates spatial geometry to jamming impact. This gap motivates our study of coordinated beamforming and control for directional jamming and nulling.

We consider a unified framework that integrates antenna beamforming with dynamics to achieve directional jamming with guaranteed nulls toward friendly clients. Specifically, we study a jammer UAV that carries two omnidirectional antennas, and which controls their phases and rotational orientation, as well as its position to i) ensure zero jamming power reaches a friendly client, and to ii)  maximize jamming power radiated towards an eavesdropper. We assume the antenna spacing is fixed and within one-half of the carrier wavelength, often referred to as a two-element uniform linear array. This configuration allows smooth beampatterns that remain robust to variations in the eavesdropper’s direction (Figure \ref{fig:beams_example}).

To achieve the first requirement, we compute a closed-form expression for the antennas' phases so that their signals cancel out in the client's direction. This forces the antennas' beampattern to take a shape similar to the right part of Figure 1. To achieve the second requirement, we formulate an optimal control problem wherein the objective is for the UAV to relocate to get in a better jamming position while conforming to actuation constraints. For this problem, we compute the optimal antenna orientation pointwise in closed form, whereas we solve for the optimal control input of the UAV using Pontryagin's principle. Our results demonstrate how integrating control and communication design can significantly enhance the efficacy of jamming operations.

\textit{Notation:} $\mathds{1}_{x\in A}(x)$ is the indicator function that is equal to $1$ when $x\in A$, and $0$ otherwise. $\norm{\cdot}$ denotes the Euclidean norm, and $\norm{\cdot}_{\infty}$ the infinity norm. The function $\textrm{atan2}(y, x):\mathbb{R}\times\mathbb{R}\rightarrow [0, 2\pi)$ denotes the two-argument arctangent. We use $I$ to denote an identity matrix of appropriate order, and $\textrm{diag}(c_1,\ldots,c_n)$ to denote a diagonal matrix with entries $c_1,\ldots,c_n$. For a vector $x$, $x_i$ denotes its $i$-th entry. 

\section{Problem Formulation}

We consider a spatial setting where a client wants to receive a wireless communication signal with carrier wavelength $\lambda > 0$, while an eavesdropper aims to intercept the same signal. We denote the position of the client in this setting as $p_c=[x_c~y_c]^\textrm{T}\in\mathbb{R}^2$, and the position of the eavesdropper as $p_e=[x_e~y_e]^\textrm{T}\in\mathbb{R}^2$. 

\subsection{Spatial Setup for Coordinated Jamming}

To prevent the eavesdropper from intercepting the communication, a UAV carries two omnidirectional antennas in order to transmit electromagnetic interference, i.e., to jam the eavesdropper. However, for the jamming to be effective, the UAV must appropriately control the antenna phases of its interference signals as well as its position. Specifically, it must control them to ensure that the interference i) arrives as powerfully as possible at the eavesdropper, hence maximizing jamming impact; and ii) cancels out at the client, hence leaving the client completely unaffected. 

We set the distance between the two antennas to be fixed and equal to $0<D\le \frac{\lambda}{2}$, so that they form a basic two-element antenna array.  Given this, if we denote the position of the UAV as $p_g=[x_g~y_g]^\textrm{T}\in\mathbb{R}^2$, then we can express the positions  $p_1, p_2\in\mathbb{R}^2$ of each of the antennas as
\begin{align*}
p_1=p_g-\frac{D}{2}(\textrm{cos}(\theta_g),\textrm{sin}(\theta_g)),\\
p_2=p_g+\frac{D}{2}(\textrm{cos}(\theta_g),\textrm{sin}(\theta_g)),
\end{align*}
where $\theta_g$ denotes the orientation of the antenna array axis relative to the reference frame at $p_g$, which we assume the UAV can mechanically adjust. We provide a geometric illustration of this setup in Figure \ref{fig:setup}. 

\begin{remark}
Instead of employing a UAV carrying two omnidirectional antennas, one could also employ two UAVs carrying one antenna each. However, this case is more challenging for precise control, antenna array calibration, and beamshape design.
\end{remark}

\begin{figure}
\centering
\includegraphics[scale=1.8]{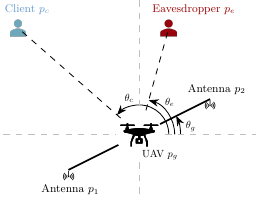}
\caption{\small A UAV positioned at $p_g=[x_g~y_g]^\textrm{T}$ carries two omnidirectional antennas. The objective is to control its position $p_g$, the angle of the two antennas $\theta_g$, as well as the antennas' phases in order to maximally jam the eavesdropper positioned at $p_e=[x_e~y_e]^\textrm{T}$. At the same time, a client positioned at $p_c=[x_c~y_c]^\textrm{T}$ should remain unaffected.} \label{fig:setup}
\end{figure}

\subsection{Jamming Beampatterns and Power}

The jamming beampattern generated by the phase selection for the two antennas should ideally cause the two transmitted signals to combine constructively at the eavesdropper $p_e$ to maximize jamming effectiveness, and destructively at the client $p_c$ to leave it unaffected. At any point $p\in\mathbb{R}^2$, this beampattern  takes the expression
\begin{equation}\label{eq:beam}
B(p; ~p_g, \theta_g, \phi_1, \phi_2) = \left|e^{j(kd_1(p)+\phi_1)}{+}e^{j(kd_2(p)+\phi_2)}\right|^2,
\end{equation}
where $\phi_1,~\phi_2\in\mathbb{R}$ are the beamsteering phases applied at the respective antennas, $d_1(p)=\norm{p-p_1},~d_2(p)=\norm{p-p_2}$ are the distances of each of the antennas from $p$, and $k=\frac{2\pi}{\lambda}$. 

To calculate jamming power, note that the beampattern $B$ can be interpreted as the directional gain of a virtual two-element antenna array with a geometric center at $p_g$. For this antenna configuration, the free-space path losses $L:\mathbb{R}\rightarrow\mathbb{R}$ at any point $p$ are
\begin{equation}\label{eq:FSPL}
L(\norm{p-p_g})=\frac{1}{4k^2\norm{p-p_g}^2},
\end{equation}
where $\norm{p-p_g}$ is the distance of the point $p$ from the antenna array's geometric center (i.e., from the UAV).
Hence, following the Friis transmission equation \cite{balanis2005antenna}, we can calculate the total jamming power arriving at any point $p$, measured in $\textrm{dBm}$, as
\begin{align}\nonumber
P(p)&=P(p; ~p_g, \theta_g, \phi_1, \phi_2)\\\nonumber=&10\textrm{log}_{10}\Big(P_0B(p; ~p_g, \theta_g, \phi_1, \phi_2) L(\norm{p-p_g})\Big) \\\nonumber =& 10\textrm{log}_{10}P_0 + 10\textrm{log}_{10}B(p; ~p_g, \theta_g, \phi_1, \phi_2) \\&+ 10\textrm{log}_{10}L(\norm{p-p_g}),  \label{eq:powerbefore}
\end{align}
where $P_0>0$ is the nominal power of each antenna in $\textrm{mW}$, and where we have assumed unit antenna gains.

\subsection{Coordinated Beamforming and Control for Maximal Jamming}

In the case of an immobile UAV, the jamming power at the eavesdropper's location $p_e$ (as given by \eqref{eq:powerbefore}) can be maximized using standard optimization tools by statically selecting the phases $\phi_1$, $\phi_2$, and the angle $\theta_g$. Here we consider a \textit{mobile} UAV, where jamming power can be further enhanced by dynamically controlling the UAV position $p_g = p_g(t)$ over time $t \ge 0$. In this setting, the jamming power \eqref{eq:powerbefore} becomes time-varying:
\begin{align}\nonumber
P(p;&~t)=P(p; ~p_g(t), \theta_g(t), \phi_1(t), \phi_2(t))  \\\nonumber =& 10\textrm{log}_{10}P_0 + 10\textrm{log}_{10}B(p; ~p_g(t), \theta_g(t), \phi_1(t), \phi_2(t)) \\ & + 10\textrm{log}_{10}L(\norm{p-p_g(t)}),  \label{eq:power}
\end{align}
and so do the values of the phases $\phi_1(t), \phi_2(t)$ and the angle $\theta_g(t)$, allowing for greater jamming flexibility. Controlling the UAV position $p_g(t)$ while concurrently selecting the antennas' transmission parameters $\theta_g(t), \phi_1(t), \phi_2(t)$ is the purpose of the remainder of the paper.

Towards this end, for all $t\ge0$, we consider the trajectory dynamics
\begin{align*}
\dot{p}_g(t)&=v_g(t),~~p_g(0)=p_{g0},\\
\dot{v}_g(t)&=u_g(t),~~v_g(0)=v_{g0},
\end{align*}
where $p_g(t),~v_g(t)\in\mathbb{R}^2$ are the UAV's position and velocity -- which coincide with the position and velocity of the center of the two-element antenna array -- and $u_g(t)\in\mathbb{R}^2$ is the UAV control input. 

Per the discussion above, maximizing the jamming power \eqref{eq:power} at the eavesdropper location $p_e$ requires carefully controlling the UAV position $p_g(t)$, the antenna orientation $\theta_g(t)$, and the phases $\phi_1(t), \phi_2(t)$ over time $t \ge 0$. Simultaneously, these parameters must be constrained so that the beampattern \eqref{eq:beam} (and thus the jamming power \eqref{eq:power}) is zero at the client location $p_c$, ensuring the client remains unaffected. To achieve this, we can solve the optimal control problem of minimizing
\begin{align}\nonumber
& J{=} \int_0^{t_f} \hspace{-1mm}\Big(\frac{1}{2}u_g^\textrm{T}(t)Ru_g(t){+}\frac{1}{2}v_g^\textrm{T}(t)Q_rv_g(t){-}a_r\sigma(P(p_e;t)) \Big)\mathrm{d}t \\\label{eq:cost}&\qquad\qquad\quad+\frac{1}{2}v_g^\textrm{T}(t_f)Q_fv_g(t_f)-a_f\sigma(P(p_e;~t_f)),\\
&\quad \textrm{s.t.}\quad B(p_c; ~t)=0, ~\forall t\in[0,t_f],\nonumber\\
&\quad\qquad \norm{u_g(t)}_{\infty}\le \bar{u}, ~\forall t\in[0,t_f],\nonumber
\end{align}
where $\bar{u} > 0$ denotes the maximum actuation capability of the UAV. Here, $a_r, a_f > 0$ are scalars that emphasize the importance of jamming impact; $R=\textrm{diag}(r_1,r_2) \succ 0$ is a weighting matrix that penalizes control effort; and $Q_r,~ Q_f \succ 0$ are weighting matrices that penalize large velocities. Finally, $\sigma:\mathbb{R} \rightarrow \mathbb{R}$ is a function that, in the nominal case of directly maximizing jamming impact, is simply the identity function $\sigma(x) = x$. However, if one wants to incentivize denial of service, $\sigma$ should be chosen as an activation function centered at the threshold for service denial. For example, in the case of denying GPS signals, $\sigma$ can be chosen as a ReLU-like function centered at $-90~\textrm{dBm}$.

\begin{remark}\label{re:2}
While the free-space path losses term $L$ indicate that smaller distances lead to higher jamming impact, it is evident from \eqref{eq:power} that such impact also depends on the beamforming function $B$. This means that the array formed by the UAV's antennas should reposition not only to get closer to the eavesdropper, but also arrive from the proper direction and with the appropriate orientation $\theta_g$. For this reason,  the position $p_g$ of the UAV is not included explicitly in \eqref{eq:cost} but, instead, its desired values are informed through the jamming impact power $P$.
\end{remark}

\begin{remark}
The formula \eqref{eq:FSPL} is valid only in the far-field, i.e., when the distance $d$ from the target is greater than $\frac{2D^2}{\lambda}$, which corresponds to a few centimeters. In practice, however, an antenna radiating a few centimeters near the eavesdropper would induce overwhelming jamming power anyway.
\end{remark}

\section{Beamforming and Control for Constrained Maximal Jamming}

Given that the cost \eqref{eq:cost} is nonlinear and subject to a hard constraint, directly minimizing it over the decision variables $u_g,~\theta_g,~\phi_1,~\phi_2$ is computationally expensive and presents a highly nonconvex landscape. To address this, we adopt a bilevel optimization approach: we first select the phases $\phi_1,~\phi_2$ to enforce the nulling constraint $B(p_c; t)=0$ at the client, and then optimize $u_g,~\theta_g$ to minimize the cost \eqref{eq:cost}.

\subsection{Phase Control for Nulling at the Client}

First, we choose the phases $\phi_1,~\phi_2$ of the jamming signals of the two antennas to enforce nulling at the client. To that end, note that the constraint $B(p_c;~ t)=0$ in \eqref{eq:cost} implies
\begin{align*}
B(p_c;~t)=0 \Longleftrightarrow &\left|e^{j(kd_1(p_c)+\phi_1)}+e^{j(kd_2(p_c)+\phi_2)}\right|^2=0 \\ \Longleftrightarrow &\left|e^{j(k(d_1(p_c)-d_2(p_c))+\phi_1-\phi_2)}+1\right|^2=0.
\end{align*}
Therefore, choosing
\begin{equation}\label{eq:phi2}
\phi_2(t)=\phi_1(t)+\pi+k(d_1(p_c)-d_2(p_c))
\end{equation}
guarantees the client remains unaffected by the jamming. Subsequently, we may choose the phase $\phi_1(t)$ arbitrarily: since it represents just an offset on the phase of the beamformed signal on the eavesdropper, it has no effect on the magnitude of the beampattern \eqref{eq:beam}. However, note that for the narrow band signal case considered here, we can adjust $\phi_1(t)$ to offset for Doppler shift $f_D$ on the jammer according to 
\begin{equation}\label{eq:phi1}
\phi_1(t)=-\int_0^t 2\pi f_D(t)\textrm{d}t{=}-\int_0^t kv_g^\textrm{T}(t)\frac{p_e-p_g(t)}{\norm{p_e-p_g(t)}}\textrm{d}t.
\end{equation}

\subsection{Orientation Control for Maximal Jamming Impact on the Eavesdropper}

Next, we design the trajectory of the antenna array's orientation, denoted by $\theta_g(t)$. This orientation influences only the jamming power $P$ in the cost function \eqref{eq:cost}, and it does so indirectly through the beamforming function $B$ defined in \eqref{eq:beam}.
We can therefore adjust the orientation $\theta_g(t)$ at each time instant $t\ge0$ to pointwise maximize \eqref{eq:beam}, and this will lead to minimization of \eqref{eq:cost}. 

To this end,  plugging the choice \eqref{eq:phi2} for the phases $\phi_1, \phi_2$ in the beampattern function \eqref{eq:beam}, we obtain
\begin{equation*}
B(p; ~p_g, \theta_g)=\left|e^{j(k(d_1(p)-d_1(p_c)-d_2(p)+d_2(p_c)))}-1\right|^2
\end{equation*}
or equivalently
\begin{multline}\label{eq:tempb1}
B(p; ~p_g, \theta_g)=2-2\textrm{cos}\Big(k(d_1(p)-d_1(p_c)\\-d_2(p)+d_2(p_c))\Big).
\end{multline}
By definition, we have $d_i(p)=\norm{p-p_i}=\norm{p-p_g+(-1)^{i+1}\frac{D}{2}(\textrm{cos}(\theta_g),\textrm{sin}(\theta_g))}$ for the distance from antenna $i\in\{1,2\}$. In addition, in the far field we have $\norm{p-p_g}\gg\frac{2D^2}{\lambda}$, and hence $d_i(p)\approx \norm{p-p_g}+(-1)^{i+1}\frac{D}{2}(\textrm{cos}(\theta_g), \textrm{sin}(\theta_g))\frac{p-p_g}{\norm{p-p_g}}$ by a first-order approximation in the law of cosines. Denoting $p-p_g=R(\textrm{cos}(\theta),\textrm{sin}(\theta))$, where $(R,\theta)$ are the polar coordinates of $p$ with respect to the reference frame centered at $p_g$ (Figure \ref{fig:setup}), this implies 
\begin{align*}
d_1(p)&\approx \norm{p-p_g}+\frac{D}{2}\textrm{cos}(\theta_g-\theta),\\
d_2(p)&\approx \norm{p-p_g}-\frac{D}{2}\textrm{cos}(\theta_g-\theta),
\end{align*}
and hence we can approximate \eqref{eq:tempb1} on the position $p=p_e$ of the eavesdropper as a function of the \textit{angular direction} $\theta_e$ of the eavesdropper (Figure \ref{fig:setup}), according to
\begin{align}\nonumber
B(\theta_e; ~p_g, \theta_g) &
\approx2{-}2\textrm{cos}(kD(\textrm{cos}(\theta_g-\theta_e)-\textrm{cos}(\theta_g-\theta_c)))\\&=2-2\textrm{cos}\left(2kD\mu~\textrm{sin}\left(\frac{\theta_e{+}\theta_c}{2}{-}\theta_g\right)\right)\hspace{-1mm}\label{eq:Btheta}
\end{align}
with 
\begin{equation*}
\mu:=\textrm{sin}\left(\frac{\theta_c-\theta_e}{2}\right).
\end{equation*}
Note that this function is purely angle-dependent, depending on the direction of the eavesdropper $\theta_e$, rather than its position $p_e$, which is the expected behavior in the far field for antenna beampatterns. 

Using the structure of the beampattern and generalizing the analysis of \cite{brian2uav}, we calculate the value of $\theta_g$ that maximizes \eqref{eq:Btheta} as follows. 

\begin{theorem}
The antenna orientation
\begin{equation}\label{eq:thetag}
\theta_g(t)=\begin{cases}\frac{\theta_c(t)+\theta_e(t)}{2}\pm\frac{\pi}{2}, ~&|\mu(t)|<\frac{\pi}{2kD}, \\ \frac{\theta_c(t)+\theta_e(t)}{2}\pm\mathrm{arcsin}\left(\frac{\pi}{2kD|\mu(t)|}\right),~ &|\mu(t)|\ge\frac{\pi}{2kD},\end{cases}
\end{equation}
is a global maximizer of the beampattern \eqref{eq:Btheta} at the direction $\theta_e$ of the eavesdropper.
\end{theorem}
\begin{proof}
First, suppose that $|\mu(t)|<\frac{\pi}{2kD}$. Then, the range of $2kD\mu~\textrm{sin}\left(\frac{\theta_e{+}\theta_c}{2}{-}\theta_g\right)$ is $[-2kD\mu, 2kD\mu]\subset[-\pi, \pi]$. In this range, the maximization of $-2\textrm{cos}\left(2kD\mu~\textrm{sin}\left(\frac{\theta_e{+}\theta_c}{2}{-}\theta_g\right)\right)$ is attained by making the cosine argument as close to $\pm\pi$ as possible, which is equivalent to making the enclosed sin argument equal to $\pm\frac{\pi}{2}$. Therefore, the global maximizer of \eqref{eq:Btheta} with respect to $\theta_g$ satisfies
\begin{equation}\label{eq:thg1}
\frac{\theta_e{+}\theta_c}{2}-\theta_g=\pm\frac{\pi}{2} \Longrightarrow \theta_g = \frac{\theta_e{+}\theta_c}{2}\pm\frac{\pi}{2}.
\end{equation}

Next, suppose $|\mu(t)|\ge\frac{\pi}{2kD}$. Then, \eqref{eq:Btheta} is maximized by 
enforcing $-2\textrm{cos}\left(2kD\mu~\textrm{sin}\left(\frac{\theta_e{+}\theta_c}{2}{-}\theta_g\right)\right)=2$, which is attained by making the cosine argument exactly equal to $\pm \pi$. This yields
\begin{multline}\label{eq:thg2}
2kD\mu~\textrm{sin}\left(\frac{\theta_e+\theta_c}{2}-\theta_g\right)=\pm \pi \\\Longrightarrow \frac{\theta_e+\theta_c}{2}-\theta_g = \pm\mathrm{arcsin}\left(\frac{\pi}{2kD\mu} \right) \\ \Longrightarrow \theta_g = \frac{\theta_e+\theta_c}{2} \pm\mathrm{arcsin}\left(\frac{\pi}{2kD\mu}\right).
\end{multline}
Combining \eqref{eq:thg1}-\eqref{eq:thg2} yields \eqref{eq:thetag}. \frQED
\end{proof}

\subsection{Position Control for Maximal Jamming Impact on the Eavesdropper}

Having chosen the jamming phases $\phi_1(t),~\phi_2(t)$ and the array orientation $\theta_g(t)$, the final step is to select the control input $u_g(t)$ that repositions the UAV for maximal jamming impact. To this end, plugging the expression of the optimal antenna orientation \eqref{eq:thetag} in \eqref{eq:Btheta}, we obtain the optimal beampattern value:
\begin{equation*}
B^\star(\theta_e;~p_g) {=} \begin{cases}2-2\textrm{cos}\left(2kD\textrm{sin}\left(\frac{\theta_c{-}\theta_e}{2}\right)\right), &|\mu(t)|<\frac{\pi}{2kD}, \\ 4, &|\mu(t)|\ge\frac{\pi}{2kD}.\end{cases}
\end{equation*}
For this beampattern, the jamming power \eqref{eq:power} at the eavesdropper's position $p_e$ takes the form
\begin{align}\nonumber
P^\star(p_e;~t)=&10\textrm{log}_{10}P_0 {+} 10\textrm{log}_{10}B^\star(\theta_e(t);~p_g(t)) \\ & + 10\textrm{log}_{10}L(\norm{p_e-p_g(t)}).\label{eq:powerstar}
\end{align}
Hence, given also that the phase design \eqref{eq:phi2} enforces the nulling constraint at the client in \eqref{eq:cost}, the cost \eqref{eq:cost} becomes
\begin{align}\nonumber
 J^\star&{=} \int_0^{t_f}\hspace{-2mm} \Big(\frac{1}{2}u_g^\textrm{T}(t)Ru_g(t){+}\frac{1}{2}v_g^\textrm{T}(t)Q_rv_g(t){-}a_r\sigma(P^\star(p_e;t)) \Big)\mathrm{d}t\\ &\qquad\quad+\frac{1}{2}v_g^\textrm{T}(t_f)Q_fv_g(t_f)-a_f\sigma(P^\star(p_e;t_f)),\label{eq:coststar}\\
 &\quad \textrm{s.t.}\quad \norm{u_g(t)}_{\infty}\le \bar{u}, ~\forall t\in[0,t_f].\nonumber
\end{align}
Given this expression of the cost, we derive its minimizing controller using Pontryagin's principle.
\begin{theorem}\label{th:bvp}
Let $u_g^\star :[0,t_f]\rightarrow\mathbb{R}^2$ be a minimizer of \eqref{eq:coststar}. Then, for $i\in\{1,2\}$:
\begin{equation}\label{eq:opt}
u_{gi}^\star(t)=\begin{cases}-r_i^{-1}\xi_{vi}(t),~ &|r_i^{-1}\xi_{vi}(t)|\le \bar{u}, \\ -\bar{u}\cdot \mathrm{sgn}(\xi_{vi}(t)),~&|r_i^{-1}\xi_{vi}(t)|>\bar{u}, \end{cases}
\end{equation}
where $\xi_p,\xi_v:[0,t_f]\rightarrow\mathbb{R}^2$ solve
\begin{equation*}
\begin{split}
\dot{\xi}_p(t) &= a_r\gamma(t)\Bigg(\frac{B_p(p_g(t))\cdot\mathds{1}_{|\mu(t)|<\frac{\pi}{2kD}}}{B^\star(\theta_e(t);~p_g(t))}{+}\frac{L_p(p_g(t))}{L(\norm{p{-}p_g(t)})}\Bigg){,}\\
\dot{\xi}_v(t) &=-\xi_p(t)- Q_rv_g(t),
\end{split}
\end{equation*}
 with boundary conditions 
\begin{equation*}
\begin{split}
{\xi}_p(t_f) &= -a_f\gamma(t_f)\Bigg(\frac{B_p(p_g(t_f))\cdot\mathds{1}_{|\mu(t_f)|<\frac{\pi}{2kD}}}{B^\star(\theta_e(t_f);p_g(t_f))}\\&\qquad\qquad\qquad\quad{+}\frac{L_p(p_g(t_f))}{L(\norm{p-p_g(t_f)})}\Bigg), \\ {\xi}_v(t_f) &=Q_f v_g(t_f).
\end{split}
\end{equation*}
Here, we have the formulas $\gamma(t)=\frac{10 \sigma'(P^\star(p_e;t))}{\mathrm{ln}10}$ and
\begin{align*}
L_p(p_g(t)) &= \frac{p_e-p_g(t)}{2k^2\norm{p_e-p_g(t)}^4},\\
B_p(p_g(t)) &= 2kD~\mathrm{sin}\left(2kD\mu(t)\right)\mathrm{cos}\left(\frac{\theta_c(t)-\theta_e(t)}{2}\right) \\ &\hspace{-10mm}\cdot \begin{bmatrix} \frac{y_c-y_g(t)}{(x_c-x_g(t))^2+(y_c-y_g(t))^2}-\frac{y_e-y_g(t)}{(x_e-x_g(t))^2+(y_e-y_g(t))^2} \\ -\frac{x_c-x_g(t)}{(x_c-x_g(t))^2+(y_c-y_g(t))^2}+\frac{x_e-x_g(t)}{(x_e-x_g(t))^2+(y_e-y_g(t))^2} \end{bmatrix}.
\end{align*}
\end{theorem}

\begin{proof}
Define the Hamiltonian of the optimal control problem of minimizing \eqref{eq:coststar} as
\begin{equation*}
H=\frac{1}{2}u_g^\textrm{T}Ru_g+\frac{1}{2}v_g^\textrm{T}Q_rv_g-a_r\sigma(P^\star)+\xi_p^\textrm{T}v_g+\xi_v^\textrm{T}u_g,
\end{equation*}
where $\xi_p,~\xi_v\in\mathbb{R}^2$ are the costates.
Since this Hamiltonian is strictly convex in $u_g$, when $|r_i^{-1}\xi_{vi}(t)|\le \bar{u}$ for $i\in\{1,2\}$, the optimal control must satisfy the stationarity condition $\frac{\partial H}{\partial u_{gi}}=0$, which yields $u_{gi}^\star(t)=-r_i^{-1}\xi_{vi}(t)$. On the other hand, when $|r_i^{-1}\xi_{vi}(t)|>\bar{u}$, Pontryagin's principle dictates that the optimal control should minimize the Hamiltonian, which yields the projected equation $u_{gi}^\star(t)=-\bar{u}\cdot \mathrm{sgn}(\xi_{vi}(t))$. Collectively, these formulas yield
\eqref{eq:opt}.

Next, from the adjoint equations, we obtain
\begin{equation}\label{eq:adj}
\begin{split}
\dot{\xi}_p(t) &= -\frac{\partial H(t)}{\partial p_g(t)} = a_r\sigma'(P^\star(t))\frac{\partial P^\star(t)}{\partial p_g(t)},\\
\dot{\xi}_v(t) &= -\frac{\partial H(t)}{\partial v_g(t)} = -\xi_p(t)- Q_rv_g(t).
\end{split}
\end{equation}
For the partial derivative of $P^\star(t)$, we have
\begin{equation}\label{eq:dP}
\begin{split}
\frac{\partial P^\star(t)}{\partial p_g(t)}&=\frac{10}{\textrm{ln}10}\frac{1}{B^\star(\theta_e(t);~p_g(t))}\frac{\partial B^\star(\theta_e(t);~p_g(t))}{\partial p_g(t)}\\&+\frac{10}{\textrm{ln}10}\frac{1}{L(\norm{p-p_g(t)})}\frac{\partial L(\norm{p-p_g(t)})}{\partial p_g(t)}.
\end{split}
\end{equation}
Here, we have
\begin{align}\label{eq:dL}
\frac{\partial L(\norm{p-p_g(t)})}{\partial p_g(t)} &= \frac{p-p_g(t)}{2k^2\norm{p-p_g(t)}^4}.
\end{align}
Moreover, following the definitions $p_g=[x_g~y_g]^\textrm{T}$, $p_e=[x_e~y_e]^\textrm{T}$, $p_c=[x_c~y_c]^\textrm{T}$, the formulas $\theta_c(t)=\textrm{atan2}(y_c-y_g(t),x_c-x_g(t))$, $\theta_e(t)=\textrm{atan2}(y_e-y_g(t),x_e-x_g(t))$, and the chain rule, when $|\mu(t)|<\frac{\pi}{2kD}$ we obtain
\begin{align}\nonumber
\hspace{-2mm}&\frac{\partial B^\star(\theta_e(t);~p_g(t))}{\partial x_g(t)}=2kD~\textrm{sin}\left(2kD\textrm{sin}\left(\frac{\theta_c(t)-\theta_e(t)}{2}\right)\right)\\\nonumber&\cdot\textrm{cos}\left(\frac{\theta_c(t)-\theta_e(t)}{2}\right)\Bigg(\frac{y_c-y_g(t)}{(x_c-x_g(t))^2+(y_c-y_g(t))^2}\\&\qquad\qquad\qquad-\frac{y_e-y_g(t)}{(x_e-x_g(t))^2+(y_e-y_g(t))^2}\Bigg)\label{eq:dBx}
\end{align}
and
\begin{align}\nonumber
&\hspace{-2mm}\frac{\partial B^\star(\theta_e(t);~p_g(t))}{\partial y_g(t)}=2kD~\textrm{sin}\left(2kD\textrm{sin}\left(\frac{\theta_c(t)-\theta_e(t)}{2}\right)\right)\\\nonumber&\cdot\textrm{cos}\left(\frac{\theta_c(t)-\theta_e(t)}{2}\right)\Bigg(-\frac{x_c-x_g(t)}{(x_c-x_g(t))^2+(y_c-y_g(t))^2}\\&\qquad\qquad\qquad+\frac{x_e-x_g(t)}{(x_e-x_g(t))^2+(y_e-y_g(t))^2}\Bigg).\label{eq:dBy}
\end{align}
On the other hand, when $|\mu(t)|\ge\frac{\pi}{2kD}$, one has $\frac{\partial B^\star(\theta_e(t);~p_g(t))}{\partial x_g(t)}=\frac{\partial B^\star(\theta_e(t);~p_g(t))}{\partial y_g(t)}=0$. Combining this with \eqref{eq:adj}-\eqref{eq:dBy} yields the flow equations of the theorem. From the transversality conditions, we also have that $\xi_v(t_f)=Q_fv_g(t_f)$ and $\xi_p(t_f)=-a_f\gamma(t_f)\frac{\partial P(t_f)}{\partial p_g(t_f)}$. We can calculate $\frac{\partial P(t_f)}{\partial p_g(t_f)}$ following a similar analysis as in the flow equations, which yields the boundary conditions of the theorem. \frQED
\end{proof}

\begin{remark}
To ensure numerical stability, the denominators $B^\star$ and $L$ should be normalized with a small positive constant. However, this regularization is not required when $\sigma$ is chosen as a ReLU-like function as discussed before Remark \ref{re:2}, and as implemented in the following numerical experiments.
\end{remark}

\section{Numerical Experiments}

\subsection{Directional Jamming Against a Static Eavesdropper}

We consider a client, located at $p_c=[3000~3000]^\textrm{T} \textrm{m}$, who wants to receive a wireless communication signal in the $1575.42~\textrm{MHz}$ frequency band. At the same time, an eavesdropper positioned at $p_e=[6000~6000]^\textrm{T} \textrm{m}$ wants to intercept the same signal. We assume that the power of the signal is $-125~\textrm{dBm}$.

\begin{figure}[!t]
\centering
\includegraphics[width=\linewidth]{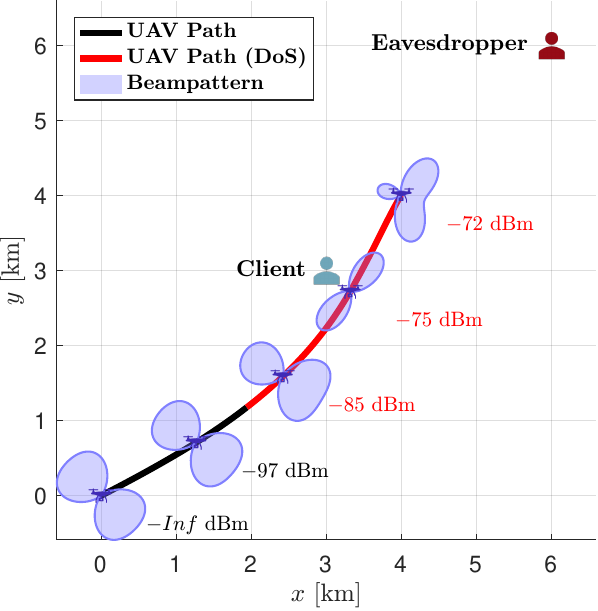}
\caption{\small The path followed by the UAV, along with snapshots of its antennas' beampattern and the corresponding jamming power attained on the eavesdropper. 
As the far-field directions to the eavesdropper and the client separate, and as the UAV gets closer to the eavesdropper, jamming power on the eavesdropper increases and eventually crosses the $-90~\textrm{dBm}$ DoS threshold.} \label{fig:path}
\end{figure}

To jam the eavesdropper and prevent it from receiving the wireless signal, we deploy a jammer UAV that carries two omnidirectional antennas separated by $D=\frac{\lambda}{2}=9.52~\textrm{cm}$ and with power $P_0=600~\textrm{mW}$ each, as shown in Figure~\ref{fig:setup}. This UAV is initially positioned at $p_g(0)=[0~0]^\textrm{T}~\textrm{m}$ and with velocity $v_g(0)=[0~0]^\textrm{T}~\textrm{m/s}$. This position, however, is suboptimal: the UAV is far from the eavesdropper, and both the client and the eavesdropper lie in the same direction in the far field. As a result, the UAV cannot effectively jam the eavesdropper without also interfering with the client. Therefore, to reposition the UAV for maximum jamming effectiveness, we solve the optimal control problem \eqref{eq:coststar} to determine its control input $u_g$, and we select the communication parameters $\phi_1,~\phi_2,~\theta_g$ according to \eqref{eq:phi2}, \eqref{eq:phi1} and \eqref{eq:thetag}. The actuation threshold for the UAV is $\bar{u}=2~\textrm{m/s}^2$.

\begin{figure}[!t]
\centering
\includegraphics[width=0.95\linewidth]{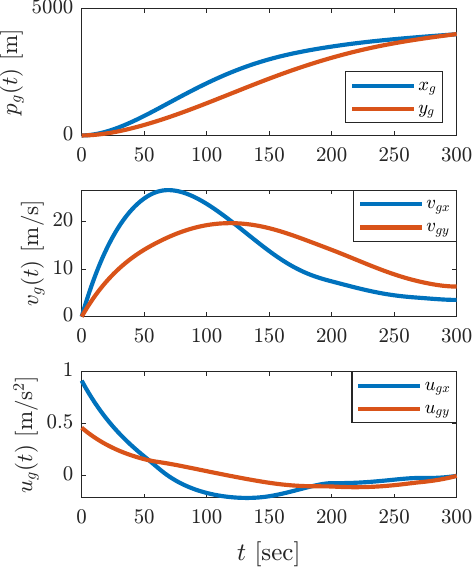}
\caption{\small The trajectories of the position $p_g(t)=[x_g(t)~y_g(t)]^\textrm{T}$, the velocity $v_g(t)=[v_{gx}(t)~v_{gy}(t)]^\textrm{T}$, and the acceleration $u_g(t)=[u_{gx}(t)~u_{gy}(t)]^\textrm{T}$ of the UAV.} \label{fig:posvel}
\end{figure}

\begin{figure}[!t]
\centering
\includegraphics[width=1\linewidth]{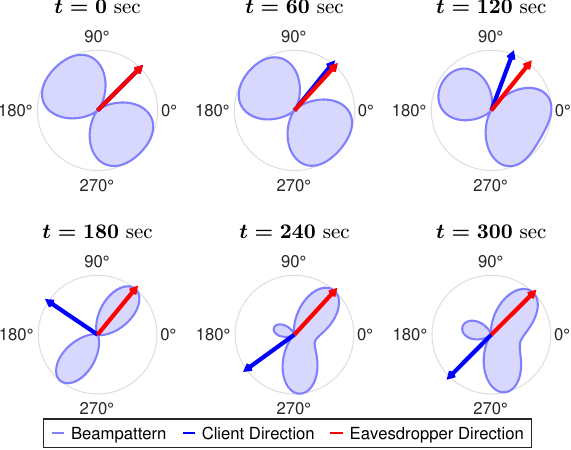}
\caption{\small The beampattern created by the UAV's antennas. As the UAV attains better angle variety between the client and the eavesdropper, the beampattern value in the direction of the eavesdropper increases. Meanwhile, the beampattern value in the client's direction remains equal to 0, i.e., no jamming power reaches the client.\vspace{-3mm}} \label{fig:beams}
\end{figure}

We choose the parameters of the optimal control problem \eqref{eq:coststar} as $t_f=300~\textrm{sec}$,   $R=\frac{200}{t_f}$, $a_r=\frac{\textrm{ln}10}{t_f}$, $Q_r=\frac{0.1}{t_f} I$ and $a_f=Q_f=0$. Moreover, since the wireless signal's strength is $-125~\textrm{dBm}$, we assume jamming interference begins becoming impactful when it crosses the $-100~\textrm{dBm}$ threshold at the receiver, and leads to denial-of-service (DoS) reliably when it crosses the $-90~\textrm{dBm}$ threshold. Given this information, we select $\sigma(\cdot)$ in \eqref{eq:coststar} as the ReLU-like function:
\begin{equation*}
\sigma(x)=\begin{cases}-100, &x<-100, \\ x,  &x\in[-100,-70], \\ -70, &x>-70, \end{cases}
\end{equation*}
and, accordingly, $\sigma'(x)=\mathds{1}_{x\in[-100,-70]}(x)$. This choice of $\sigma$ motivates the optimal $u_g^\star$ to generate trajectories that degrade or deny service at the eavesdropper. At the same time, it avoids giving excessive rewards for UAV paths that pass too close to the eavesdropper and yield jamming powers above $-70~\textrm{dBm}$ (as the DoS outcome remains unchanged). Moreover, it preserves numerical stability by zeroing the flow equations in Theorem~\ref{th:bvp} whenever either $B^\star$ or $\norm{p - p_g}$ is close to zero. Note, to improve numerical stability, we use the smooth approximation $\sigma'(x)\approx \frac{1}{1 + e^{-200(x +100)}} - \frac{1}{1 + e^{-200(x +70)}}$.

\begin{figure}[!t]
\centering
\includegraphics[width=0.95\linewidth]{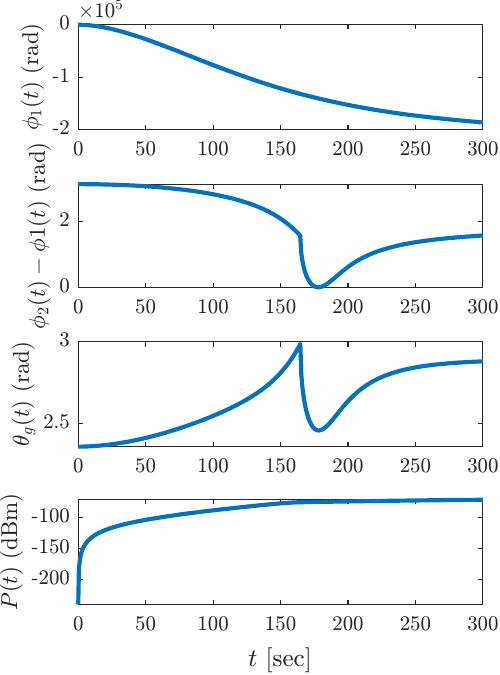}
\caption{\small The trajectories of the phases $\phi_1(t),~\phi_2(t)$ and the angle $\theta_g(t)$ of the UAV's antenna, as well as the total received jamming power $P(t)$ on the eavesdropper. \vspace{-3mm}} \label{fig:coms}
\end{figure}

We solve the optimality equations of Theorem \ref{th:bvp} using the shooting-based solver bvp4c in MATLAB. We initialize the solver with a path that drives the UAV to the midpoint between the eavesdropper and the client. 
Figure \ref{fig:path} shows the resulting optimal path of the UAV as well as the optimal beampattern of its antennas and the corresponding jamming powers attained on the eavesdropper. We notice that the UAV follows a curved path in order to attain better angle variety between the eavesdropper and the client, i.e., in order to increase the difference between $\theta_c(t)$ and $\theta_e(t)$. This angular separation increases the magnitude of $|\mu(t)|$, thereby enhancing the directional jamming gain $B^\star(\theta_e)$ toward the eavesdropper. As the UAV approaches the eavesdropper, free-space path losses also decrease, and the received jamming power surpasses the $-90~\textrm{dBm}$ threshold, leading to DoS on the eavesdropper. Figure~\ref{fig:posvel} further shows the UAV's position, velocity, and acceleration profiles, which remain physically realistic and respect the actuation constraint of $2~\textrm{m/s}^2$.

Figure~\ref{fig:beams} shows snapshots of the antenna beampatterns, with the directions to the eavesdropper and the client indicated in red and blue, respectively. Initially, these directions coincide, and the directional gain toward the eavesdropper is set to zero to avoid interfering with the client. As the UAV moves and the angular separation between the two directions increases, the beampattern is shaped to enhance the gain toward the eavesdropper, which eventually reaches its maximum value of $4$. Figure~\ref{fig:coms} depicts the evolution of the communication parameters $\phi_1(t)$, $\phi_2(t)$, and $\theta_g(t)$ that shape the beampattern, as well as the resulting jamming power $P(t)$ received by the eavesdropper. A visible nonsmooth transition occurs around $t = 160~\textrm{sec}$, corresponding to the time when the expression for $B^\star$ saturates at its maximum value of $4$.

\begin{figure}[!t]
\centering
\includegraphics[width=1\linewidth]{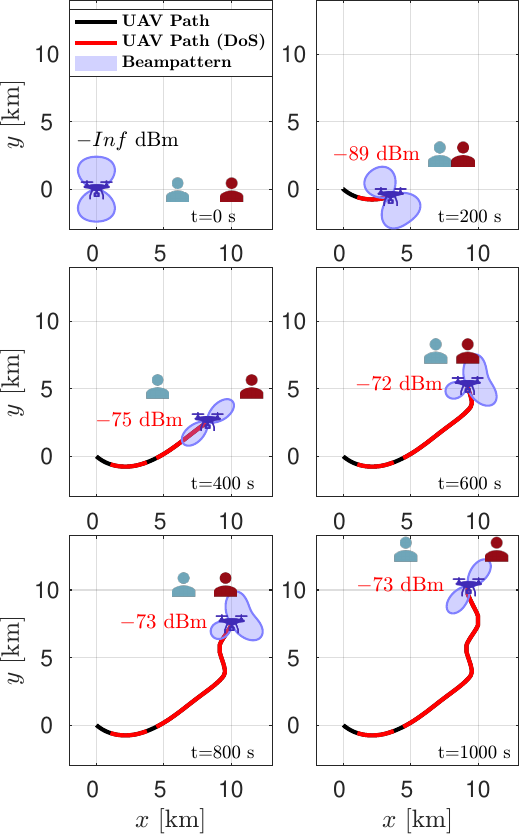}
\caption{\small The path followed by the UAV in the receding horizon jamming implementation, with the client denoted in steel blue and the eavesdropper in dark red. The UAV initially moves south to get a better angle separation between the client and the eavesdropper, which allows the beampattern \eqref{eq:Btheta} to increase. Subsequently, it moves between the client and the eavesdropper and follows them, achieving maximal angle separation while minimizing path losses. \vspace{-3mm}} \label{fig:path_move}
\end{figure}

\subsection{Receding Horizon Implementation Against a Moving Eavesdropper}

Subsequently, we apply the scheme from the previous subsection to a scenario where both the client and the eavesdropper are in motion.
We assume the client and eavesdropper are initially located at $[6000~0]^\mathrm{T}~\mathrm{m}$ and $[10000~0]^\mathrm{T}~\mathrm{m}$, respectively, and follow northbound trajectories. To account for their changing positions, the UAV replans its trajectory and beamforming every $2$ seconds,  for a total of $1000$ seconds (i.e., 500 replanning instances). Moreover, to ensure that no jamming power reaches the client due to its movement between replanning intervals, the UAV adjusts the beamforming online by recomputing the optimal antenna angle $\theta_g(t)$ from the formula \eqref{eq:thetag} and the phase from \eqref{eq:phi2}. The rest of the parameters are as previously. Note, each replanning problem took approximately $15~\textrm{ms}$ to solve.

Figure \ref{fig:path_move} shows snapshots of the resulting path of the UAV and its beampattern, along with the positions of the client and eavesdropper. We notice that the UAV initially moves to the east-southeast, despite the client and eavesdropper moving north. This is because a southeastern trajectory allows the UAV to get better angle separation between the client's and the eavesdropper's direction, hence allowing the beampattern $B^\star$ on the eavesdropper to increase. Subsequently, the UAV converges to a path between the client and the eavesdropper which maximizes angle separation -- and hence $B^\star$ on the eavesdropper -- while reducing path losses. Finally, jamming power crosses the $-100~\textrm{dBm}$ interference threshold after $42.6$ seconds, and the $-90~\textrm{dBm}$ DoS threshold after $88.6$ seconds (with a short drop to $-90.4~\textrm{dBm}$ at $t=220~\textrm{sec}$).

\section{Conclusion}

We consider the coordinated beamforming–control problem for a UAV tasked with jamming towards an eavesdropper while leaving a friendly client unaffected. The UAV carries two omnidirectional antennas and can control their phases, orientation, and its own position. We derive a closed-form expression for the antenna phases that enables nulling at the client, and determine the antennas’ orientation and the UAV trajectory that maximizes jamming at the eavesdropper by solving an optimal control problem.

Future work includes extending the proposed jamming scheme to a game-theoretic framework that anticipates a mobile eavesdropper’s worst-case future trajectory.

\balance

\bibliographystyle{ieeetr}      
\bibliography{references.bib}

\end{document}